\documentclass[twocolumn, superscriptaddress, showpacs, pra]{revtex4-1}

\usepackage{color,amsbsy,amssymb,latexsym,amsfonts, amsmath,cancel}
\usepackage{braket}
\usepackage{mathrsfs}
\usepackage{graphicx}

\newtheorem{proposition}{Proposition}
\newtheorem{theorem}{Theorem}

\newenvironment{proof}[1][Proof]{\textbf{#1.} }{\ \rule{0.5em}{0.5em}}

\newcommand{\bel}[1]{\begin{equation}\label{#1}}                     
\newcommand{\bal}[1]{\begin{eqnarray}\label{#1}}   
\newcommand{\be}{\begin{equation}}               
\newcommand{\ba}{\begin{eqnarray}}           
\newcommand{\ee}{\end{equation}}
\newcommand{\ea}{\end{eqnarray}}

\newcommand{\abs}[1]{\left| #1 \right|}
\newcommand{\orddif}[2]{\frac{d #1}{d #2}}
\newcommand{\bea}{\begin{equation}}
\newcommand{\eea}{\end{equation}}

\begin{document}
\begin{flushright}
OCU-PHYS 362
\end{flushright}
\title{Equation of State for the One-Dimensional Attractive $\delta$-Potential Bose Gas\\ in the Weak-Coupling Regime}

\author{Tsubasa Ichikawa}
\email[Email:]{ tsubasa@qo.phys.gakushuin.ac.jp}
\altaffiliation[present address:]{ Department of Physics, Gakushuin University, Tokyo 171-8588, Japan}
\affiliation{Research Center for Quantum Computing, Interdisciplinary Graduate School of Science and Engineering, Kinki University, 3-4-1 Kowakae, Higashi-Osaka, 577-8502, Japan}

\author{Izumi Tsutsui}
\email[Email:]{izumi.tsutsui@kek.jp}
\affiliation{Theory Center, Institute of Particle and Nuclear Studies,
High Energy Accelerator Research Organization (KEK), 1-1 Oho, Tsukuba, Ibaraki 305-0801, Japan}

\author{Nobuhiro Yonezawa}
\email[Email:]{yonezawa@sci.osaka-cu.ac.jp}
\affiliation{Osaka City University Advanced Mathematical Institute (OCAMI), 3-3-138, Sugimoto, Sumiyoshi-ku, Osaka, 558-8585, Japan}

\date{\today}
\begin{abstract}
Approximated formulas for real quasimomentum and the associated energy spectrum are presented 
for one-dimensional Bose gas with weak attractive contact interactions.  On the 
basis of the energy spectrum, we obtain the equation of state in the high-temperature
region, which is found to be the van der Waals equation without volume 
correction.
\end{abstract}

\maketitle

\paragraph*{Introduction.}

One-dimensional quantum gas has been attracting considerable attention
since its seminal experimental demonstrations by using cold atoms 
\cite{Kinoshita_Wenger_Weiss_2004, Paredes_Widera_Murg_et_al_2004}.
To date, one of the bosonic gas models extensively studied is the Lieb-Liniger model 
\cite{Lieb_Liniger_1963}, which is solvable and characterized by a tunable coupling 
constant.

In the case of repulsive interaction, much literature is found both for the zero-temperature region \cite{Lieb_Liniger_1963, Girardeau_1960_0303,Girardeau_0610016, Oelkers_Batchelor_Bortz_Guan_0511694} and the finite-temperature region \cite{Yang_Yang_1969,Yang_1970,
 1751-8121-44-10-102001, PhysRevLett.100.090402}. 
The Yang-Yang equations \cite{Yang_Yang_1969}, 
 whose solutions describe the thermodynamical behavior of the system, have been demonstrated experimentally \cite{PhysRevLett.100.090402}. Generalization to anyonic systems is also found in \cite{Gutkin_1987_0515, Kundu_9811247, Batchelor_Guan_Kundu_08051770}.
In parallel, the case of attractive interaction 
has been investigated \cite{McGuire_1964, 1998PhRvA..57.3317M, PhysRevLett.98.150403, 2007JSMTE..08...32C, 2011JMP....52l2106P, PhysRevA.67.013608, PhysRevA.68.043619, PhysRevA.76.063620} in the last decade, which uncovered
various exotic  features including bound states with complex quasimomentum \cite{Lieb_Liniger_1963, McGuire_1964, 1998PhRvA..57.3317M, PhysRevLett.98.150403, 2007JSMTE..08...32C, PhysRevA.76.063620}. 
Phase structure with respect to the coupling strength is found in the paradigm of Gross-Pitaevskii mean-field theory \cite{PhysRevA.67.013608, PhysRevA.68.043619}.

In this Brief Report, we present two simple but notable results about the Lieb-Liniger 
model valid for {\it weak} attractive interactions.  One of them is the explicit form of the
quasimomentum and energy spectrum, and the other is the equation of state in the high-temperature region which resembles 
the van der Waals equation of state \cite{v_d_Waals_1873}.

\paragraph*{Energy Spectrum in the Weak Coupling Limit.}

The Lieb-Liniger model \cite{Lieb_Liniger_1963} describes an $N$-partite 
bosonic system in one-dimensional space with point interactions.  It is governed by the Hamiltonian
\begin{align}
H =-\frac{\hbar^2}{2m}\sum_{i=1}^N \frac{\partial^2}{\partial x^2_i}
-c \sum_{i < j }\delta(x_i-x_j),
\end{align}
where $m$ is the mass of the particle, $c \ge 0$ is the coupling constant 
of the attractive $\delta$ interaction, and the variable $x_i\in[0,L]$ represents the coordinate 
of the $i$th particle.   The eigenfunctions $\psi$ are 
symmetric $\psi(\cdots, x_i,\cdots,x_j,\cdots)=\psi(\cdots, x_j,\cdots,x_i,\cdots)$ and obey
the periodicity conditions $\psi(\cdots, x_i,\cdots)=\psi(\cdots, x_i+L,\cdots)$ for all $i$.    For convenience we
hereafter work with the unit $\hbar=m=1$.

Following the standard treatment of the model, we adopt the Bethe ansatz,
\begin{align}
\psi(x)=\sum_{ \sigma} a_\sigma \exp \left( i \sum_{i=1}^N k_{ \sigma (i)} x_i \right)
\label{ba}
\end{align}
for the simplex region $0<x_1<x_2<\cdots<x_N<L$. 
Here $k_{i}$ is a quasi-momentum (rapidity) and the summation is over all permutations $ \sigma$ of the particles $i \mapsto  \sigma (i)$ on which the coefficients $a_\sigma$ depend.
Plugging (\ref{ba}) into the Schr\"odinger equation $H\psi=E\psi$, we obtain the  Bethe equations 
\cite{Lieb_Liniger_1963}
\begin{align}
L k_i &= (2 n_i+N-1)\pi
    + 2\sum_{j \neq i} \arctan\delta_{i,j}.
    \label{eq:bethe_eq}
\end{align}
Here, $n_i$ is an integer,  $\delta_{i,j}$ is
the rescaled relative rapidity defined by
\be
\delta_{i,j}=(k_i-k_j)/c
\ee
for $c \ne 0$, and  the arctangent takes the principal value, $\abs{\arctan\delta_{i,j}}\le\pi/2$.
The energy $E$ then reads $E = \sum_i k_i^2/2$.
 
In what follows, we consider the integer set $\{n_i\}$ so that
$
k_i\in\mathbb{R}
$
for all $i$ for the (dimensionless) weak coupling,
\be
\epsilon:=cL/\pi\ll1.
\label{ke}
\ee
We can confine ourselves to the case of the ordering, 
\be
k_1<k_2<\cdots<k_N,
\label{ordering}
\ee
by relabeling the indices $i$ appropriately.
Note that (\ref{ordering}) implies $\delta_{i,j}>0$ for $i>j$.  
We also assume that in the weak regime $k_i$ are all 
regular with respect to $c$. 
 
Two remarks are in order.  First, there actually exist $\{n_i\}$ for few-partite systems 
on which the condition (\ref{ke}) is satisfied. 
For example, it is shown in \cite{Lieb_Liniger_1963} that $k_i\in\mathbb{R}$ 
for bipartite systems for $\epsilon<4/\pi$. For tripartite systems, the same 
holds in the weak-coupling regime if the system has nonzero  
$\{n_i\}$ at the noninteracting limit $c \to 0$ \cite{1998PhRvA..57.3317M}.
Second, under the regularity assumption the ordering (\ref{ordering}) implies no level crossing in the weak regime.

We now show the following proposition.
\begin{proposition}
\label{edn}
For $\epsilon \ll1$, 
the rescaled relative rapidity and the coupling have the tradeoff relation,
\begin{align}
\epsilon\delta_{i,j}&=2( \bar{n}_i-\bar{n}_j)+{\cal O}(\epsilon)
\label{kd}
\end{align}
with
\be
\bar n_i:=n_i+i-1.
\label{knbar}
\ee
\end{proposition}
\begin{proof}
To show this, we first note that
for $c\to0$, 
\be
\arctan\delta_{i,j}\to\begin{cases}
\pi/2      & \text{for $i>j$}, \\
-\pi/2      & \text{for $i<j$}.
\end{cases}
\ee
Substituting this to the Bethe equations (\ref{eq:bethe_eq}), we find 
\be
k_i\to\frac{2\pi}{L}\bar n_i,
\ee
with $\bar n_i$ given by  (\ref{knbar}).
This implies that the number $\bar n_i$ is nothing but the quantum number of the free case $c=0$. 
Since (\ref{ordering}) holds for $c=0$, for the weak coupling regime the regularity condition assures that  
\be
\bar n_1<\bar n_2<\cdots<\bar n_N,
\label{nbar}
\ee
and also that
\be
k_i=\frac{2\pi}{L}\bar n_i+\frac{1}{L}{\cal O}(\epsilon),
\ee
on account of ${\cal O}(c)=(1/L){\cal O}(\epsilon)$ from (\ref{ke}).
We then obtain
\begin{align}
\epsilon\delta_{i+1.i}&=\frac{L}{\pi} (k_{i+1} -k_i)=2( \bar{n}_i-\bar{n}_j)+{\cal O}(\epsilon),
\label{lkk}
\end{align}
which completes the proof.
\end{proof}

The tradeoff relation (\ref{kd}) implies
$
|\delta_{i,j}| = 2(\bar n_i-\bar n_j)/\epsilon+{\cal O}(1)\gg1
$
or
$
1/\abs{\delta_{i,j}}={\cal O}(\epsilon),
$
which is useful to approximate $k_i$.   Indeed, with 
\begin{align}
\arctan x=\begin{cases}
-1/x+\pi/2+{\cal O}(1/x^3),      & x>1, \\
-1/x-\pi/2+{\cal O}(1/x^3),      & x<-1,
\end{cases}
\end{align}
substituted for Eq.~(\ref{eq:bethe_eq}), we find
\begin{align}
k_i = \frac{2\pi}{L}\bar{n}_i-\frac{2}{L} \sum_{j \neq i}^N \frac{1}{\delta_{i,j}} +  
\frac{1}{L}\mathcal{O}( \epsilon^3).
\label{eq:pre_approx_spectrum}
\end{align}
Utilizing (\ref{kd}) again, we can eliminate 
$\delta_{i,j}$ from (\ref{eq:pre_approx_spectrum}) to arrive at the following theorem.
\begin{theorem}
\label{kln}
For $\epsilon\ll1$, the rapidity is approximated by
\begin{align}
k_i =\frac{2\pi}{L}\bar{n}_i- \frac{\epsilon}{L}
\sum_{j \neq i}^N \frac{1}{\bar{n}_i-\bar{n}_j}+\frac{1}{L}{\cal O}(\epsilon^2).
  \label{leq}
\end{align}
\end{theorem}
This is one of the announced results of this paper.

Let us now derive the approximated energy spectrum $E$ to the order of ${\cal O}(\epsilon)$.
Summing up all the squared rapidities (\ref{leq}), we obtain
\begin{align}
E_{\bar{n}_1,\bar{n}_2,\ldots,\bar{n}_N}&=\sum_{i=1}^N\frac{k_i^2}{2}\nonumber\\
 &\approx \sum_{i=1}^N\frac{2\pi^2 \bar{n}_i^2 }{ L^2}
    - 2\pi\frac{\epsilon}{L^2}
       \sum_{i=1}^N\sum_{j \neq i} \frac{\bar{n}_i}{\bar{n}_i-\bar{n}_j} \nonumber\\
&= \sum_{i=1}^N\frac{2\pi^2 \bar{n}_i^2 }{ L^2}
    -c\frac{N(N-1)}{L},
\label{energy}
\end{align}
where we have used
\begin{align}
2 \sum_{i=1}^N \sum_{j \neq i} \frac{\bar n_i}{\bar n_i-\bar n_j}&=
\sum_{i=1}^N \sum_{j \neq i} \frac{\bar n_i}{\bar n_i-\bar n_j}+
\sum_{j=1}^N \sum_{i \neq j} \frac{\bar n_j}{\bar n_j-\bar n_i}\nonumber\\
&=\sum_{i=1}^N \sum_{j \neq i} \frac{\bar n_i- \bar n_j}{\bar n_i-\bar n_j}\nonumber\\
&= N(N-1).
\end{align}

The energy spectrum (\ref{energy}) is additive for the particles, and
we may rewrite it as
$
E_{\bar n_1, \bar n_2,\cdots\bar n_N}\approx\sum_{i=1}^NE_{\bar n_i}
$
with
\be
E_{\bar n_i}=\frac{2\pi^2}{L^2}\bar n_i^2-c\frac{N-1}{L}.
\label{eq:separated_spectrum}
\ee
The system thus behaves as an assembly of  non-interacting particles 
of the energy $E_{\bar n}$, consisting of 
the kinetic part $2\pi^2 \bar{n}^2/L^2$ and the averaged potential part $-c(N-1)/L$
which is proportional to the number density $(N-1)/L$. 

\paragraph*{Equation of State.}

Let us consider the equation of state of the system for the high-temperature region:
\begin{align}
\beta:= 1/k_{\rm B}T \ll L^2/2\pi^2,
\label{highT}
\end{align}
where $k_{\rm B}$ is the Boltzmann constant and $T$ is the temperature.
Under the Maxwell-Boltzmann distribution, 
the partition function $Z$ becomes
\begin{align}
Z
    &\approx \int_{[-\infty,\infty]^N} \prod_{i=1}^N  d p_i \,  e^{-\beta E_{p_1,p_2,\ldots,p_N}}\nonumber\\
&\approx
L^N\left(2 \pi\beta\right)^{-N/2}\exp\frac{\beta c N^2}{L}.
\end{align}
To obtain this, we approximated the sum over  
$\{\bar n_i\}$ by the integral over $p_i := 2\pi \bar n_i/L$, ignoring the case of complex 
rapidities  ({\it i.e.}, $p_i=p_j$ for some $i,j$) which is measure zero.

\begin{figure}[t]
	\includegraphics[width=3in]{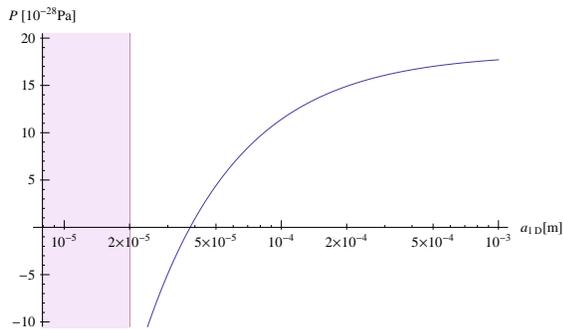}
   \caption{(Color online)
   Pressure of one-dimensional Cs gas as a function of the one-dimensional
   scattering length $a_{\rm1D}$, evaluated for $N=25$, $L=30\,\mu\text{m}$, $T=0.16\,\text{nK}$.   For lower values of $a_{\rm1D}$, 
   we find that $P$ may become negative, which occurs before reaching the shaded range $a_{\rm1D}<2\times10\,\text{$\mu$m}$ where the weak coupling condition 
   (\ref{ke}) is no longer valid.
   }
   \label{PTa}
 \end{figure}

Since the volume of the system is given by $L$ in our one dimensional model, the pressure $P$ is evaluated as
$
P=\frac{1}{\beta}\orddif{}{L}\ln Z. 
$
We then obtain the following theorem.
\begin{theorem}
\label{vdW}
The equation of state in the high-temperature region reads
\begin{align}
\left(P+c\frac{N^2}{L^2}\right)L   \approx Nk_{\rm B}T.
\label{VdW}
\end{align}
\end{theorem}
This is in fact the van der Waals equation without the volume correction.   The absence of the  volume correction is understood by the fact that the Lieb-Liniger model
adopts only point particles.  


The one-dimensional system of bosons with attractive interactions of our interest has actually been realized as the 
super-Tonks-Girardeau (sTG) phase of an ultracold gas of cesium (Cs) atoms
using the confinement-induced resonance \cite{Haller09}.  In this experiment,    
one observes a multiple of elongated tubes, each consisting of approximately 25 atoms along the size of $3\times10\,\mu$m on average.
To analyze this in our context, we take these numbers for $N$ and $L$, respectively, with the mass
$m\approx2.2\times10^{-25}\,\text{kg}$.  From (\ref{highT}), we then find that our equation of state (\ref{VdW}) is valid for 
$T\gg2\pi^2\hbar^2/mL^2k_{\rm B}\approx8\times10\,\text{pK}$.   
On the other hand, according to \cite{Olshanii98} our coupling $c$ is related to the  tunable 
one-dimensional scattering length $a_{\rm 1D}>0$ by $c=2\hbar^2/ma_{\rm 1D}$.   Thus our weak coupling condition (\ref{ke}) 
implies $a_{\rm 1D}\gg2L/\pi \approx2\times10\,\mu\text{m}$.
In contrast, we find $a_{\rm 1D}\approx3\times10^2\,{\rm nm}$ in the experiment \cite{Haller09}, which is way out of the valid range.

If we are allowed to extend the range by tuning the parameters of experiment such as the strength of applied magnetic field appropriately without spoiling the sTG phase up to
the valid range, then our formula suggests an intriguing possibility.    This is seen in the behavior of pressure $P$ (see 
Fig.~\ref{PTa}) which becomes negative for lower values of $a_{\rm 1D}$.  This takes place when the pressure correction term surmounts the ideal gas term in (\ref{VdW}), suggesting a phase transition there.
This may also be realized, {\it e.g.}, by increasing the number density $N/L$ even if $a_{\rm 1D}$ is fixed.

\paragraph*{Conclusion and Discussions.}

In this paper we obtained the approximated energy spectrum of the one-dimensional
attractive Bose gas where particles interact weakly with each other via the contact $\delta$ potential. 
The coupling $\epsilon$ and the rescaled 
relative rapidity $\delta_{i,j}$ fulfill the tradeoff relation (Proposition \ref{edn}). 
As a result, we obtained the closed expression of the approximated rapidity (Theorem \ref{kln}).
Based on this, we found that the equation of state is the van der Waals equation without the volume correction (Theorem \ref{vdW}).  

The induced correction term is proportional to the attractive coupling $c$ with the squared number density, which is  precisely the same as the standard 
correction.  This correction is conventionally 
ascribed to the average interparticle potentials in the classical thermodynamics derivation, whereas in our quantum-mechanical model the effect is taken care of by the nontrivial boundary conditions of the $\delta$ potential imposed at the positions of the particles.     The valid range of our analysis does not seem to overlap with 
the parameter range of the present experiment with ultracold Cs gas, but if it can be extended we may observe a novel phase transition as a consequence of the quantum boundary effect.

In closing, we mention that the $\delta$ potential is not the only possible point interactions admitted in quantum mechanics.    Indeed, it has been known that in one dimension 
we have a U(2) family of point interactions each characterized by distinct 
boundary conditions \cite{AGHH} including the $\delta$ potential as a special case.   Despite that all of them are equally zero-range, these potentials can in principle give rise to different spectra through the non-trivial 
boundary conditions, bearing a variety of palpable physical effects as pointed out in \cite{Fulop_Miyazaki_Tsutsui_0309095,Fulop_Tsutsui_0612011,Fulop_Tsutsui_0906_2626,Yonezawa_0904_1134}.
We may therefore expect similar exotic outcomes to appear, for instance, in the equation of state when such potentials other than the $\delta$ are considered.

\begin{acknowledgments}
T.I. was supported by \lq\lq Open Research Center\rq\rq~Project for
 Private Universities, matching fund subsidy, MEXT, Japan.
The research of N.Y. was supported in part by a Grant-in-Aid for Scientific Research (Grant No. 2054278) as well as JSPS Bilateral Joint Projects (JSPS-RFBR collaboration) from the MEXT.
N.Y. was also partially supported by the JSPS Institutional Program for
Young Researcher Overseas Visits
\lq\lq Promoting international young researchers in mathematics and
mathematical sciences led by OCAMI.\rq\rq
\end{acknowledgments}

\end{document}